\documentclass[conference,a4paper]{IEEEtran}
\usepackage[utf8]{inputenc} 
\usepackage[T1]{fontenc}
\usepackage{url}
\usepackage{ifthen}
\usepackage{cite}
\usepackage[cmex10]{amsmath} % Use the [cmex10] option to ensure complicance
\usepackage{flushend}

\usepackage{listings}
\usepackage{amsmath}
\usepackage{amsthm}
\usepackage{tikz}
\usepackage{caption}
\usepackage{array}
\usepackage{mdwmath}
\usepackage{multirow}
\usepackage{mdwtab}
\usepackage{eqparbox}
\usepackage{amsfonts}
\usepackage{tikz}
\usepackage{multirow,bigstrut,threeparttable}
\usepackage{amsthm}
\usepackage{array}
\usepackage{bbm}
\usepackage{subcaption}
\usepackage{epstopdf}
\usepackage{mdwmath}
\usepackage{mdwtab}
\usepackage{eqparbox}
\usepackage{tikz}
\usepackage{latexsym}
\usepackage{cite}
\usepackage{amssymb}
\usepackage{bm}
\usepackage{amssymb}
\usepackage{graphicx}
\usepackage{mathrsfs}
\usepackage{epsfig}
\usepackage{psfrag}
\usepackage{setspace}
\usepackage{hyperref}
\usepackage{algorithm}
\usepackage{algpseudocode}
\usepackage{stfloats}
\usepackage{mathtools}
\usepackage{algorithm}
\usepackage{algorithmicx}
\usepackage{algpseudocode}
\usepackage{float}

\usepackage{epsfig}
\usepackage{ytableau}
\usepackage{hypcap}

\newtheorem{theorem}{Theorem}

\newcommand{\beq}{\begin{equation}}
\newcommand{\eeq}{\end{equation}}
\newcommand{\mc}{\mathcal}
\newcommand{\mb}{\mathbb}

\newcommand{\eqdef}{\triangleq}
\newcommand{\ra}{\rightarrow}

\newcommand{\rev}{^\text{r}}
\newcommand{\B}[1]{\bar{#1}}

\DeclareMathOperator*{\argmin}{argmin}

\begin{document}

%\title{Holding Cost}
%\title{Optimally holding a bag: an information-theoretic analysis}
%A Markov chain I-projection
%A Markov chain I-projection: fixed stationary distribution and the cost of holding a bag
%Optimal bag-holding: a thermodynamic analysis
%Markov chains with nearest
%The nearest Markov chain with a fixed stationary distribution: I-projection and thermodynamic analysis
%Closest Markov chain with a fixed stationary distribution
\title{Minimum Power to Maintain a Nonequilibrium Distribution of a Markov Chain}

% The thermodynamics of maintaining a nonequilibrium stationary distribution
% 
% Optimal control of Markov chains
% 
% The cost of maintaining a non-equilibrium distribution
% 
% Work to run a controlled Markov chain naturally measured by KL divergence [thermodynamics motivations]
% 
% [who has looked at this before
% question has been looked at in [refs]
% we give physics perspective, motivation, apply results to particular question, "holding" cost

%\author{Dmitri S. Pavlichin, Yihui Quek, and Tsachy Weissman}

\author{%
  \IEEEauthorblockN{Dmitri S. Pavlichin}
  \IEEEauthorblockA{Stanford University\\
                    dmitrip@stanford.edu}
  \and
  \IEEEauthorblockN{Yihui Quek}
  \IEEEauthorblockA{Stanford University\\
                    yquek@stanford.edu}
  \and
  \IEEEauthorblockN{Tsachy Weissman}
  \IEEEauthorblockA{Stanford University\\
                    tsachy@stanford.edu}
}

\date{\today}

\maketitle

\begin{abstract}

Biological systems use energy to maintain non-equilibrium distributions for long times, e.g.  of chemical concentrations or protein conformations.  What are the fundamental limits of the power used to ``hold'' a stochastic system in a desired distribution over states?  We study the setting of an uncontrolled Markov chain $Q$ altered into a controlled chain $P$ having a desired stationary distribution.  Thermodynamics considerations lead to an appropriately defined Kullback-Leibler (KL) divergence rate $D(P||Q)$ as the cost of control, a setting introduced by Todorov, corresponding to a Markov decision process with mean log loss action cost.

The optimal controlled chain $P^*$ minimizes the KL divergence rate $D(\cdot||Q)$ subject to a stationary distribution constraint, and the minimal KL divergence rate lower bounds the power used.  While this optimization problem is familiar from the large deviations literature, we offer a novel interpretation as a minimum ``holding cost'' and compute the minimizer $P^*$ more explicitly than previously available.  We state a version of our results for both discrete- and continuous-time Markov chains, and find nice expressions for the important case of a reversible uncontrolled chain $Q$, for a two-state chain, and for birth-and-death processes.

% We further formulate the related setting of a fixed time horizon and a 
% time-varying controlled chain.

\end{abstract}

% \begin{abstract}
% 
% Biological systems use energy to maintain non-equilibrium distributions for 
% long times (e.g.  of chemical concentrations or protein shapes).  What are the 
% fundamental limits of the power used to ``hold'' a stochastic system in a 
% desired distribution of states?  We study the setting of an uncontrolled Markov 
% chain $\und{P}$ altered into a controlled chain $P$ having a desired stationary 
% distribution.  Thermodynamics considerations lead to an appropriately defined 
% Kullback-Leibler (KL) divergence $D(P||\und{P})$ as the cost of control, a 
% setting introduced by Todorov; equivalently, we study Markov decision processes 
% with log loss action cost.
% 
% The optimal controller $P^*$ minimizes the KL divergence $D(\cdot||\und{P})$ 
% subject to a stationary distribution constraint, and the minimal KL divergence 
% lower bounds the power used.  While this optimization problem is familiar from 
% the large deviations literature, we offer a novel interpretation as a minimum 
% ``holding cost'' and compute the minimizer $P^*$ more explicitly than 
% previously available.  We state a version of our results for both discrete- and 
% continuous-time Markov chains, and find nice expressions in the case -- 
% important in thermodynamics -- of a reversible uncontrolled chain $\und{P}$.  
% 
% % We further formulate the related setting of a fixed time horizon and a 
% % time-varying controlled chain.
% 
% \end{abstract}

\section{Introduction} \label{sec:introduction}

Let $(X_1,X_2,\ldots)$ be a sequence of random variables forming a first-order 
%Let $(X_t)_{t \in \{0,1,\ldots\}}$ be a sequence of random variables forming a 
%first-order 
Markov chain on a finite set $\mc{X}$ with
%irreducible, aperiodic 
transition probability 
%probability %matrix $Q$, where $Q_{ij} = \mbf{P}_Q(X_{t+1}=j|X_t=i)$
matrix $Q$, where $Q_{ij} = \mb{P}_Q(X_{t+1}=j|X_t=i)$
%$P_0 = (P_0[i,j])_{i,j\in\mathcal{X}}$, where $P_0[i,j] \eqdef 
%\mb{P}(X_{t+1}=j|X_t=i)$ 
for all $t$ and $i,j\in\mc{X}$.%, so that
%, where the subscript of $0$ marks $P_0$ as the ``uncontrolled'' or ``base'' 
%chain.  
We think of $Q$ as the ``uncontrolled'' or ``base'' chain.
If $Q$ is irreducible and aperiodic, then there exists a unique stationary 
distribution $\pi$, viewed as a row vector:
%Let $\mc{M}$ denote the set of probability distributions on $\mc{X}$ viewed as 
%row vectors, and let
%
%Let $\pi_0$ denote the unique stationary distribution of the chain $P_0$ 
%viewed as a row vector:
\begin{equation}
    %\pi_0 P_0 = \pi_0
    \pi Q = \pi
    %\pi^{(Q)} Q = \pi^{(Q)}
    %\pi_Q Q = \pi_Q
    %\und{\pi} \und{P} = \und{\pi}
\end{equation}
Let the initial state $X_1$ have distribution $\mu$.  Then, denoting $X_1^T 
\eqdef (X_1, \ldots X_T)$, the distribution of $X_1^T$ is:
%If $X_0$ is drawn from $\pi^{(Q)}$, then, denoting by $X_0^T \eqdef (X_0, \ldots 
%X_T)$,
%$\pi_i = \mb{P}_Q(X_t=i)$ for all $t$ and $i \in \mc{X}$
%$\pi = (\mb{P}(X_t = i))_{i \in \mc{X}}$ for all $t$ and
%\begin{equation}
%    %\mb{P}_Q(X_0,\ldots,X_T) = \pi_{X_0} \prod_{t=1}^T Q_{X_{t-1},X_t}
%    %\mathbf{P}_Q(X_0,\ldots,X_T) = \pi^{(Q)}_{X_0} \prod_{t=1}^T 
%    %Q_{X_{t-1},X_t}
%    \mathbf{P}_Q(X_0^T = x_0^T) = \pi_Q[x_0] \prod_{t=1}^T Q[x_{t-1},x_t]
%\end{equation}
\begin{align}
    p_{Q,\mu}(x_1^T) &\eqdef \mb{P}_{Q,\mu}(X_1^T = x_1^T) = \mu_{x_1} 
    \prod_{t=2}^T Q_{x_{t-1},x_t} \label{eq:def_MC_distribution}
\end{align}

%\begin{align}
%    p_Q(x_0^T) &\eqdef \mb{P}_Q\{X_0^T = x_0^T\} = \pi^{(Q)}_{x_0} \prod_{t=1}^T 
%    Q_{x_{t-1},x_t} \label{eq:def_MC_distribution}
%\end{align}

% \begin{align}
%     p_Q(x_1^T|x_0) &\eqdef \mb{P}_Q\{X_1^T = x_1^T|X_0=x_0\} = \prod_{t=1}^T 
%     Q_{x_{t-1},x_t} \label{eq:def_MC_distribution}
% \end{align}

%\begin{equation}
%    D(\mbf{P}_P||\mbf{P}_Q)
%\end{equation}

Let $\pi^*$ be some probability distribution on $\mc{X}$ viewed as a row 
vector.  We study the nearest Markov chain transition matrix $P$ to $Q$ having 
$\pi^*$ as its stationary distribution:
\begin{align}
    %P_1^* = \argmin_{P_1: \pi_1 P_1 = \pi_1} d(P_1,P_0) 
    %P^* &= \argmin_{P: \pi_P = \pi^*} D(P||Q) 
    P^* &= \argmin_{P: \pi^* P = \pi^*} D(P||Q) 
    \label{eq:optimization_problem_in_introduction}
\end{align}
where $D(P||Q)$ is the Kullback-Leibler (KL) divergence rate between Markov 
chains with transition matrices $P$ and $Q$ 
\cite{CsiszarCoverChoi87}~\footnote{This expression is independent of $\mu = 
\mb{P}(X_1)$ for aperiodic, irreducible $P$.}:
\begin{equation}
    %D(P||Q) \eqdef \lim_{T \ra \infty} \frac{1}{T} \mb{E}_P 
    %\log\left(\frac{\mb{P}_P(X_0,\ldots,X_T)}{\mb{P}_Q(X_0,\ldots,X_T)}\right)
    %D(P||Q) \eqdef \lim_{T \ra \infty} \frac{1}{T} \mbf{E}_P 
    %\log\left(\frac{\mbf{P}_P(X_0,\ldots,X_T)}{\mbf{P}_Q(X_0,\ldots,X_T)}\right) 
    %\label{eq:def_KL_div_MC_discrete_time_in_introduction}
    D(P||Q) \eqdef \lim_{T \ra \infty} \frac{1}{T} \mb{E}_P 
    \log\left(\frac{p_P(X_0^T)}{p_Q(X_0^T)}\right) 
    \label{eq:def_KL_div_MC_discrete_time_in_introduction}
\end{equation}
where we note that $D(P||Q)$ is independent of the initial distribution $\mu$.

%where $d$ is a measure of distance between transition matrices and the minimum 
%is over transition matrices $P$ having the desired stationary distribution.  
We think of $P$ as the ``controlled'' or ``driven'' chain and $D(P||Q)$ as the 
cost of control per unit time -- the power.
%More generally, we may consider a target set $\varPi^*$ of distributions on 
%$\mc{X}$ and restrict optimization 
%(\ref{eq:optimization_problem_in_introduction}) to chains $P$ whose stationary 
%distribution is in $\varPi^*$.
%such that $\pi_P \in \varPi^*$.
% %The subscript $1$ marks $P_1$ as the ``controlled'' or ``driven'' chain and 
% %$d(P_1, P_0)$ is the cost of control.  We study the case of We study the 
% case
% %We consider
% $d(P,Q) = D(P||Q)$, where $D$ is the Kullback-Leibler (KL) divergence 
% appropriately defined for Markov chains.  
We further consider the analogous question for the case of a continuous-time 
Markov chain $(X_t)_{t \in [0,\infty)}$.

%We further consider the case of a continuous 

This setting is inspired by the following thought experiment due to Feynman 
\cite{FeynmanLectures}: a person holds a heavy bag above the floor for an hour 
and gets tired.  The net work done on the object is zero\footnote{Since work is 
the product of force and distance displaced, the latter of which is zero.}, so 
why does she get tired?  A table could hold the same bag indefinitely without 
an energy source, and so could the person if she were frozen solid, 
%unable to use her muscles, 
the bag hanging on her stiff, lifeless limb.
The latter observation implicates the microscopic dynamics of muscles as key to 
this question.
%As the latter observation suggests, the microscopic dynamics of muscles as key 
%to this question, the answer being roughly the following [cite a fuller 
%discussion?]: 
A toy model for the motion of striated muscle fibers -- see 
\cite{Huxley57,Qian00} for an extended discussion -- is of a random walk in a 
periodic energy potential\footnote{The myosin protein joining the fibers 
together is doing the random walking, the energy potential having 
periodically-spaced minima corresponding to 
%[conformations of the protein] 
discrete steps along the fiber.}.  Attachment of the heavy bag pulls on the 
muscle fibers, biasing the random walk in the direction of gravity, tilting the 
energy potential.  The person must use chemical energy\footnote{The hydrolysis 
of ATP molecules.} to de-bias the random walk in such a way that the bag is 
held at the desired height above the floor.  
%The uncontrolled random walk at the microscopic level is important to this 
%story: 
If the person is frozen solid,
%or is a table, 
then the underlying random walk stops\footnote{Or slows down a lot, as lowering 
the temperature reduces the transition rate between the potential's energy 
minima.}, and so chemical energy is no longer used to hold the bag.
%suspend a frozen arm.  The above treatment glosses over many details, see 
%[Huxley57,Qian00]
We ask: what is a lower bound on the power to hold the bag?
%We ask: what minimum power to hold the bag does thermodynamics demand?
%
%We ask: what is the minimum power to hold the bag demanded by thermodynamics?
%
%We ask: what lower bound on the power used to hold the bag does thermodynamics 
%set?

%We distill the above story to the optimization problem 
%(\ref{eq:optimization_problem_in_introduction}).  
We recast the above story as optimization problem 
(\ref{eq:optimization_problem_in_introduction}).
The state space $\mc{X}$ corresponds to the possible configurations of the 
system (the position of myosin along a fiber and its internal state).
%the possible configurations of myosin (its position along a fiber and its 
%``internal'' ``state'').  
The uncontrolled Markov chain $Q$ corresponds to the underlying fluctuations of 
the myosin molecule along a filament and the controlled chain $P$ corresponds 
to chemically driving the system.  

The control goal is macroscopic: the net force the person exerts on the bag is 
the sum of the forces due to each microscopic subsystem (myosin protein).  We 
do not get to control each subsystem separately, but can control them all in 
the same way, so that each subsystem corresponds to a trajectory drawn from 
Markov chain $P$ independently of other subsystems.  This notion of macroscopic 
control is reflected in the KL divergence cost 
%Thus the KL divergence cost 
(\ref{eq:def_KL_div_MC_discrete_time_in_introduction}), which is stated in 
terms of the probability distribution (\ref{eq:def_MC_distribution}) over 
microscopic trajectories, rather than in terms of a single trajectory.  
%[stationarity, average across space, time.  control distribution rather than 
%single trajectory] 

Our choice of the Kullback-Leibler divergence as the control cost function is 
motivated by this quantity's appearance in thermodynamics as proportional to 
the free energy difference from the equilibrium distribution over trajectories, 
which in turn lower bounds the work to prepare a non-equilibrium distribution 
over trajectories.
KL divergence control (the microscopic per-trajectory setting) was introduced 
in the reinforcement learning literature by Todorov \cite{Todorov06,Todorov09} 
and has connections to data compression; we discuss this and other motivations 
for our work in section \ref{sec:KL_div_cost}. The problem of maintaining a 
target nonequilibrium distribution has been studied recently by 
\cite{HorowitzZhouEngland17,HorowitzEngland17}, using a different notion of 
cost.  We discuss our work's relation to prior works in section \ref{subsec: 
comparison}.

Minimizing the KL divergence with respect to the first argument --  computing 
the I-projection -- 
%yields a large deviations interpretation to the minimum-cost controlled chain.
connects this setting to large deviations theory.  Indeed, most of the 
minimum-cost controlled chains we compute first appeared in the computation of 
rate functions for large deviations of the empirical measure of both discrete- 
and continuous-time Markov chains 
\cite{DemboZeitouniBook,BaldiPiccioni98,CsiszarCoverChoi87}.  The novelty of 
our work lies in relating these results to the minimum-power control setting -- 
showing the minimized KL divergence to be the minimum power to ``hold'' a 
nonequilibrium distribution; in aggregating related problem statements -- in 
continuous and discrete time and for reversible base chains; in computing some 
of these minimizations more explicitly than previously available; and in 
computing these minimizations for a few common examples like the 
birth-and-death chain and the two-state chain.

This work is organized as follows.  Section \ref{sec:KL_div_cost} motivates the 
use of KL divergence as energy cost by drawing from information theory and 
optimization settings and contains definitions of this cost function in 
discrete and continuous time. Section \ref{sec:minimum_power} shows how this 
energy cost of holding a given target distribution may be analytically 
minimized.  Section \ref{sec:examples} contains several examples that apply our 
theory to calculate the minimum-power controlled chain, including a 
birth-and-death chain which serves as a toy model of the muscular fiber, 
addressing the motivating question of Feynman. We conclude with a summary and 
outlook in section \ref{sec:discussion}.

We release code for computing the minimum-cost chains in this work at 
\href{https://github.com/dmitrip/controlledMC}{https://github.com/dmitrip/controlledMC}.

\section{Kullback-Leibler divergence rate as the cost of control of Markov 
chains} \label{sec:KL_div_cost}

We motivate the KL divergence rate between Markov chains in both discrete and 
continuous time as the cost function lower bounding the power in the 
bag-holding thought experiment.  We present a thermodynamics perspective in 
subsection \ref{subsec:thermo} and an equivalent perspective due to Todorov 
\cite{Todorov06,Todorov09} of a Markov decision process with log ratio cost 
function in subsection \ref{subsec:log_loss_action_cost}.  We summarize known 
expressions for the KL divergence between Markov chains in discrete and 
continuous time in subsection \ref{sec:MC_discrete_time} and 
\ref{sec:MC_continuous_time} respectively. Finally, \ref{subsec: comparison} 
places this work in context with related work. 

\subsection{KL divergence in thermodynamics}\label{subsec:thermo}

We summarize briefly the appearance of the KL divergence in measuring work in 
statistical mechanics. 
%[and data compression, large deviations, gambling]
Below, let $D(p||q) = \sum_{i \in \mc{X}} p(i) \log(p(i)/q(i))$ denote the 
Kullback-Leibler (KL) divergence between distributions $p$ and $q$ on finite 
set $\mc{X}$ and let $H(p) = \sum_{i \in \mc{X}} p(i) \log(1/p(i))$ denote the 
entropy of distribution $p$. 
%and let $H(p) = \sum_i p_i \log(1/p_i)$ denote the entropy of distribution 
%$p$.
%\subsection{Physical systems}
%\subsubsection{Independent samples}
Let \begin{equation}
    q = \left(\frac{1}{Z}e^{-\beta U(i)}\right)_{i \in \mc{X}} 
    %q = \left(e^{-\beta (U(i) - F_\beta)}\right)_{i \in \mc{X}} 
    \label{eq:Boltzmann_distribution}
\end{equation}
be the Boltzmann distribution on $\mc{X}$, where $U(\cdot)$
%, in units of Joules, 
is the energy function (sometimes called the energy potential or internal 
energy), $\beta$ is the inverse temperature, and $Z = \sum_{i \in \mc{X}} e^{-\beta U(i)}$ is the partition function.
%and $F_\beta$ is the free energy 
%determined by the normalization condition $\sum_{i \in \mc{X}} q_i = 1$.
Denote the free energy of distribution $p$ by $F(p)$ 
\cite{MezardMontanariBook}:
\begin{align}
    %F(p) \eqdef \mb{E}_p[U_X] - \frac{1}{\beta} H(p)
    %G(p) &\eqdef \mb{E}_p(U(X)) - \frac{1}{\beta} H(p) \\
    %&= -\frac{1}{\beta} \sum_{i \in \mc{X}} p(i) \big(\log(q(i)) - \beta 
    %F_\beta + \log(p(i))\big) \\
    F(p) &\eqdef \mb{E}_p(U(X)) - \frac{1}{\beta} H(p) \\
    %&= -\frac{1}{\beta} \sum_{i \in \mc{X}} p(i) \big(\log(q(i)) + \log(Z) + \log(p(i))\big) \\
    &= \frac{1}{\beta} \mb{E}_p \log 
    \left(\frac{p_X}{q_X}\right) - \frac{1}{\beta} \log(Z)\\
    %&= -\frac{1}{\beta} \log(Z) + \frac{1}{\beta} D(p||q)
    &= \frac{1}{\beta} D(p||q) - \frac{1}{\beta} \log(Z)
\end{align}
where the expectation $\mb{E}_p(\cdot)$ is over random variable $X$ with 
distribution $p$.  Then $F(\cdot)$ is minimized at equilibrium $p=q$, so that 
$F(q) = -\frac{1}{\beta}\log(Z)$.

%Thermodynamics tells us \cite{KittelKroemerBook} that the work $W(p)$ to 
%prepare the ``controlled'' distribution $p$ starting from the ``base'', 
%equilibrium distribution $q$ (also known as the ``work on the system'') is at 
%least the free energy difference:
In thermodynamics \cite{KittelKroemerBook,LandauLifshitz} the work $W$ to 
prepare the ``controlled'' distribution $p$ starting from the ``base'', 
equilibrium distribution $q$ (also known as the work on the system) is at least 
the free energy difference:
%In thermodynamics, the work to prepare the ``controlled'' distribution $p$ 
%starting from the ``base'', equilibrium distribution $q$ (also known as the 
%work on the system) is at least the free energy difference:
\begin{equation}
    %W(p) \geq F(p) - F(q) = \frac{1}{\beta} \mb{E}_p \log 
    %\left(\frac{p(X)}{q(X)}\right) = \frac{1}{\beta} D(p||q) 
    %\label{eq:work_lower_bound_free_energy}
    W \geq F(p) - F(q) = \frac{1}{\beta} D(p||q) 
    \label{eq:work_lower_bound_free_energy}
\end{equation}
% If we are interested in preparing a distribution from set $\mc{P}$, rather 
%than a particular distribution $p$, then the work $W(\mc{P})$ to prepare some 
%distribution from this set is lower bounded:
%\begin{equation}
%    W(\mc{P}) \eqdef \inf_{p \in \mc{P}} W(p) \geq \inf_{p \in \mc{P}} 
%    \frac{1}{\beta} D(p||q_\text{B}) \label{eq:work_lower_bound_minimization}
%\end{equation}
%As is customary in this setting, hidden in the background is a notion of a 
%time process transforming an initial into a final state
As is customary in this setting, there is in the background a notion of a 
stochastic process transforming initial states into final states, and the work 
$W$ in (\ref{eq:work_lower_bound_free_energy}) is an average over realizations 
of this process.  In Appendix (\ref{app:physical_systems}) we provide a 
physical example in the spirit of the Szilard's engine thought experiment in 
thermodynamics, for which the KL divergence does emerge as the work done. 
%In the bag-holding thought experiment, there is an implicit averaging of the 
%work done on the many copies of the myosin system, 
In the bag-holding thought experiment, there is a large collection of 
independent myosin systems, and the total work is the sum of the works on each 
system.  We imagine the number of microscopic systems to be large enough that
fluctuations about the average work $W$ per subsystem are small, so it is this 
average work that's our object of study.

The KL divergence cost is familiar in data compression, where the ``energy'' of 
symbol $X$ drawn from distribution $p$ is $U(X) = -\log(p_X)$.  Any compression 
scheme must use at least $\mb{E}_p(U(X)) = H(p)$ bits  to encode a sample from 
$p$ \cite{Shannon1948,CoverThomasBook} on average over draws from distribution 
$p$.  If we use a compression scheme that instead uses $-\log(q_X)$ bits to 
encode symbol $X$ -- a mismatched code -- then we would pay $D(p||q)$ extra 
bits per symbol on average.  Section \ref{sec:MC_discrete_time} contains 
analogous remarks for compressing samples drawn from Markov chain 
distributions.

\subsection{Markov chains in discrete time} \label{sec:MC_discrete_time}

In the Markov chain control setting, we apply the preceding picture with 
alphabet $\mc{X}^{T}$ (trajectories of length $T$) and with Markov chain 
distributions on $\mc{X}^T$ with a desired marginal distribution $\pi^*$.  We 
consider the continuous time setting in section \ref{sec:MC_continuous_time}.

A discrete time Markov chain distribution $p_{Q,\mu}^{(T)}$ on the set 
$\mc{X}^{T}$ is the Boltzmann distribution with energy function 
$U_{Q,\mu}^{(T)}$ parametrized by the stochastic transition matrix $Q = 
(Q_{ij})_{i,j\in\mc{X}}$ and the initial distribution vector $\mu = (\mu)_{i 
\in \mc{X}}$ (obtained by taking the logarithm of the rightmost quantity in 
(\ref{eq:def_MC_distribution})):  
%Let $x_1^T \eqdef (x_1, \ldots, x_T)$, then
\begin{equation}
    %U_{Q,\mu}(x_1^n) = -\log(\mu_{x_1}) - \sum_{i=2}^n \log(Q_{x_{i-1},x_i})
    U_{Q,\mu}^{(T)}(x_1^T) = -\log(\mu_{x_1}) - \sum_{t=2}^{T} 
    \log(Q_{x_{t-1},x_t})
    %U_{Q,\mu,T}(x_1^T) = -\log(\mu_{x_1}) - \sum_{t=2}^{T} 
    %\log(Q_{x_{t-1},x_t})
    %3U(x_1^T; Q, \mu) = -\log(\mu_{x_1}) - \sum_{t=2}^{T} 
    %\log(Q_{x_{t-1},x_t})
\end{equation}
where $x_1^T \eqdef (x_1, \ldots, x_T)$.  Then $p_{Q,\mu}^{(T)}(x_1^T) = e^{-U_{Q,\mu}^{(T)}(x_1^T)}$, where $Z = \beta = 
1$ (\ref{eq:Boltzmann_distribution}).  Given another transition matrix $P$ and 
initial distribution $\nu$, the work to prepare Markov chain distribution 
$p_{P,\nu}^{(T)}$ starting from the Boltzmann distribution $p_{Q,\mu}^{(T)}$ is 
lower bounded by the free energy difference $F(p_{P,\nu}^{(T)}) - 
F(p_{Q,\mu}^{(T)})$ (\ref{eq:work_lower_bound_free_energy}).  In the limit $T 
\ra \infty$, the work per time step -- the power -- is lower bounded by
\begin{align}
    \lim_{T \ra \infty} \frac{1}{T} W(p_{P,\nu}^{(T)}) &\geq \lim_{T \ra 
    \infty} \frac{1}{T} \left(F(p_{P,\nu}^{(T)}) - F(p_{Q,\mu}^{(T)})\right) 
    \label{eq:power_lower_bound_free_energy_mc} \\
    &= \lim_{T \ra \infty} \frac{1}{T} D(p_{P,\nu}^{(T)}||p_{Q,\mu}^{(T)}) \\
    %&= \lim_{T \ra \infty} \frac{1}{T} \mb{E}_{p_{P,\nu}^{(T)}} 
    %\log\left(\frac{p_{P,\nu}^{(T)}(X_1^T)}{p_{Q,\nu}^{(T)}(X_1^T)}\right) \\
    &= \lim_{T \ra \infty} \frac{1}{T} \mb{E}_{p_{P,\nu}^{(T)}} 
    \log\left(\frac{p_{P,\nu}^{(T)}(X_1,\ldots,X_T)}{p_{Q,\nu}^{(T)}(X_1,\ldots,X_T)}\right) 
    \label{eq:KL_div_MC_discrete_time_expression} \\
    &= \sum_{i \in \mc{X}} \pi(P)_i \sum_{j \in \mc{X}} 
    P_{ij} \log\left(\frac{P_{ij}}{Q_{ij}}\right) \\
    &\eqdef D(P||Q) \label{eq:def_KL_div_MC_discrete_time_1}
\end{align}
where $\pi(P)$ is the stationary distribution of transition matrix $P$ and the 
last equality defines the KL divergence rate \cite{CsiszarCoverChoi87} between 
Markov chains with transition matrices $P$ and $Q$.

In data compression, a sample $X_1^T \sim p_{Q,\mu}^{(T)}$ can be compressed on 
average to at least $T H(Q)$ bits \cite{CoverThomasBook}, where $H(Q)$ is the 
entropy rate of Markov chain with transition matrix $Q$:
\begin{equation}
    H(Q) = -\sum_{i \in \mc{X}} \pi(Q)_i \sum_{j \in \mc{X}} Q_{ij} 
    \log(Q_{ij})
\end{equation}
Encoding samples from distribution $p_{Q,\mu}^{(T)}$ with respect to a 
mismatched code based on distribution $p_{P,\nu}^{(T)}$ incurs an average cost 
per unit time of at least $D(P||Q)$ extra bits.

Maximizing the lower bound on power (\ref{eq:power_lower_bound_free_energy_mc}) 
over transition matrices $P$ with target stationary distribution $\pi^*$ yields 
the optimization problem (\ref{eq:optimization_problem_in_introduction}):
%on the power to prepare a Markov chain distribution on trajectories with 
%target stationary distribution $\pi^*$ yields the optimization problem 
%(\ref{eq:optimization_problem_in_introduction}):
\begin{equation}
    P^* = \argmin_{P: \pi^* P = \pi^*} D(P||Q) 
    \label{eq:optimization_problem_in_MC}
\end{equation}

\subsection{Log loss action cost} \label{subsec:log_loss_action_cost}

Another path to the same optimization problem 
(\ref{eq:optimization_problem_in_MC}) (KL divergence as a lower bound on the 
work to maintain a nonequilibrium distribution) is in terms of a Markov 
decision process with log loss action cost, a setting introduced by Todorov 
\cite{Todorov06,Todorov09}.  Let $Q$ be the uncontrolled chain and let $P$ be 
the controlled chain.  Let $c(i, j, P, Q)$ be the \textit{microscopic cost} 
paid when a transition is made from $X_t = i$ to $X_{t+1} = j$ when the 
controller chooses transition probability matrix $P$.  KL divergence control 
amounts to using the log likelihood ratio:
\begin{equation}
    %c(i, j, P) \eqdef \log\left(\frac{P[i,j]}{Q[i,j]}\right)
    c(i, j, P, Q) \eqdef \log\left(\frac{P_{ij}}{Q_{ij}}\right)
\end{equation}
If we view the rows of $P$ and $Q$ as Boltzmann distributions with different 
energy potentials -- that is, if we choose energy functions $E_i(\cdot)$, 
$E'_i(\cdot)$ such that $P_{ij} = e^{-E_i(j)}$ and $Q_{ij} = e^{-E'_i(j)}$ -- 
then the microscopic cost is the difference in energies: $c(i,j,P,Q) = E'_i(j) 
- E_i(j)$.

%If $X_t \sim \mu$, then denote the expected cost at time $t$ by $D_\mu(P||Q)$:
If $X_t \sim \mu$, then let the \textit{cost} $D_\mu(P||Q)$
%$C(\mu, P, Q)$ 
be the expected microscopic cost (the average cost paid per microscopic 
system):
%the expected cost paid at time $t$ is denoted by $D_\mu(P\|Q)$:
\begin{align}
    %D_\mu(P||Q) &\eqdef \mb{E}_{\mu, P}(c(X_t, X_{t+1}, P)) \\
    %C(\mu, P, Q) &\eqdef \mb{E}_{\mu, P}(c(X_t, X_{t+1}, P, Q)) \\
    D_\mu(P||Q) &\eqdef \mb{E}_{\mu, P}(c(X_t, X_{t+1}, P, Q)) \\
    %&= \sum_{i \in \mc{X}} \mu[i] \sum_{j \in \mc{X}} P[i,j]\log 
    %\left(\frac{P[i,j]}{Q[i,j]}\right)
    &= \sum_{i \in \mc{X}} \mu_i \sum_{j \in \mc{X}} P_{ij}\log 
    \left(\frac{P_{ij}}{Q_{ij}}\right) \label{eq:cost_disc_time_KL_div}
\end{align}
Thus 
%$C(\mu,P,Q)$ 
$D_\mu(P||Q)$ is a $\mu$-weighted KL divergence between the rows of transition 
matrices $P$ and $Q$.  We are interested in macroscopic control of $\mu_t$, 
rather than microscopic control of $X_t$, so our setup differs from the setting 
introduced by \cite{Todorov06,Todorov09}: we average the control cost over 
$\mu_t$, so there is no randomness in our setting.  Finally for irreducible, 
aperiodic transition matrix $P$, we have the identity
\begin{equation}
    %D(P||Q) \eqdef D_{\pi(P)}(P||Q) \label{eq:def_KL_div_MC_discrete_time}
    %D(P||Q) \eqdef C(\pi(P),P,Q) \label{eq:def_KL_div_MC_discrete_time}
    %D(P||Q) \eqdef D_{\pi(P)}(P||Q) \label{eq:def_KL_div_MC_discrete_time}
    D_{\pi(P)}(P||Q) = D(P||Q) \label{eq:def_KL_div_MC_discrete_time}
\end{equation}
where $\pi(P)$ is the stationary distribution of $P$.  Minimizing the cost 
$D_{\pi^*}(P||Q)$ with respect to $P$ such that $\pi^* P = \pi^*$ is 
optimization problem (\ref{eq:optimization_problem_in_MC}).

\subsection{Markov chains in continuous time} \label{sec:MC_continuous_time}

The setup of section \ref{sec:MC_discrete_time} has a natural counterpart for 
continuous-time Markov chains.  Let $\bar{Q} = (\bar{Q}_{ij})_{i,j \in \mc{X}}$ 
denote the transition rate matrix of the uncontrolled continuous-time Markov 
chain $(X_t)_{t \in \mb{R}_{\geq 0}}$, where henceforth the overbar notation 
corresponds to rate matrices.  Let $\bar{P}$ be the controlled rate matrix and 
let $X_t \sim \mu_t$.  Then
\begin{equation}
    \frac{d}{dt} \mu_t = \left(\frac{d}{dt} \mb{P}_{\bar{P}}(X_t = i)\right)_{i 
    \in \mc{X}} = \mu_t \bar{P} \Rightarrow \mu_t = \mu_0 e^{\bar{P} t} 
    \label{eq:controlled_state_continuous_time}
\end{equation}
where $e^{\bar{P}t}$ denotes the matrix exponential.  Note that every rate 
matrix $\bar{P}$ satisfies $\bar{P}_{ij} \geq 0$ for $i \neq j$ and 
$\bar{P}_{ii} = -\sum_{j\in\mc{X}:j\neq i} \bar{P}_{ij} \leq 0$, so the row 
sums of $\bar{P}$ are 0.  Conditioned on $X_t = i$, the time until the next 
jump is exponentially-distributed with a mean of $-1/\bar{P}_{ii}$, and the 
probability to jump to $j$ is proportional to $\bar{P}_{ij}$ for $i \neq j$.

The natural notion of KL divergence rate $D(\bar{P}||\bar{Q})$ between 
transition rate matrices $\bar{P}$ and $\bar{Q}$ is 
\cite{Kesidis93,BaldiPiccioni98,deLaFortelle00, BertiniFaggionatoGabrielli14} 
the limiting log likelihood ratio, analogous to 
(\ref{eq:KL_div_MC_discrete_time_expression}):
\begin{align}
    D(\bar{P}\|\bar{Q}) &\eqdef \lim_{T \ra \infty} \frac{1}{T} 
    \mb{E}_{p_{\bar{P},\mu_0}^{(T)}} 
    \log\left(\frac{p_{\bar{P},\mu_0}^{(T)}((X_t)_{t \in 
    [0,T]})}{p_{\bar{Q},\nu_0}^{(T)}((X_t)_{t \in [0,T]})}\right) \\
    &= \sum_{i \in \mc{X}} \pi(\bar{P})_i \sum_{j \in \mc{X}: j \neq i} 
    \left(\bar{Q}_{ij} - \bar{P}_{ij} + \bar{P}_{ij} 
    \log\left(\frac{\bar{P}_{ij}}{\bar{Q}_{ij}}\right)\right) 
    \label{eq:KL_div_MC_continuous_time_expression}
\end{align}
where $p_{\bar{P},\mu_0}^{(T)}(X_0^T)$ denotes the likelihood under rate matrix 
$\bar{P}$ and initial distribution $\mu_0$, and where $\pi(\bar{P})$ is the 
stationary distribution of rate matrix $\bar{P}$~\footnote{That is, 
$\pi(\bar{P}) \bar{P} = (0)_{i \in \mc{X}}$.  Equivalently, $\pi(\bar{P}) 
e^{\bar{P} t} = \pi(\bar{P})$ for all $t \geq 0$.}.  The quantity in the second 
summation in (\ref{eq:KL_div_MC_continuous_time_expression}) is the KL 
divergence between two Poisson distributions with means $\bar{P}_{ij}$ and 
$\bar{Q}_{ij}$.

The optimization problem analogous to (\ref{eq:optimization_problem_in_MC}) in 
continuous time is:
\begin{equation}
    \bar{P}^* = \argmin_{\bar{P}: \pi(\bar{P}) = \pi^*} D(\bar{P}||\bar{Q}) 
    \label{eq:optimization_problem_in_MC_continuous_time}
\end{equation}

\subsection{Comparison to prior work} \label{subsec: comparison}

Recent work \cite{HorowitzZhouEngland17,HorowitzEngland17} considers the 
question of the minimum power used to maintain a nonequilibrium state.  Their 
setting uses a different notion of cost than we do and also makes some 
restrictions about the base and controlled chains (they work in continuous 
time, assume that the base chain $\bar{Q}$ is reversible, and only allow 
controlled chains $\bar{P}$ such that $\bar{P}_{ij} \geq \bar{Q}_{ij}$ for all 
$i \neq j$ -- this corresponds to the biochemical mechanism of adding transitions with non-negative rates).  \cite{HorowitzZhouEngland17,HorowitzEngland17} minimize the 
entropy production rate among all controlled chains with the desired target 
distribution $\pi^*$ and find that ``fast control is optimal'': there is in 
general no optimally controlled chain, but given any chain $\bar{P}$ that has 
the target distribution $\pi^*$, we can come arbitrarily close to the minimum 
entropy production bound by speeding up $\bar{P}$ arbitrarily 
much\footnote{That is, using $c \bar{P} = (c \bar{P}_{ij})_{i,j\in\mc{X}}$ as 
the controlled chain and letting $c \ra \infty$. This corresponds to the 
statement in \cite{HorowitzZhouEngland17} than the ``added edges (should) 
operate much faster than the equilibrium transitions''.} (while incurring an 
arbitrarily large KL divergence cost $D(\bar{P}||\bar{Q})$ according to our 
metric).
%The very fast controlled chain of 
%\cite{HorowitzZhouEngland17,HorowitzEngland17} has a very large KL divergence 
%cost.  

The difference between the two notions of cost -- KL divergence rate in this 
work and entropy production rate in 
\cite{HorowitzZhouEngland17,HorowitzEngland17} is the difference between total 
energy used and the efficiency with which that energy is used as measured by 
entropy production rate.  The very fast controlled chain of 
\cite{HorowitzZhouEngland17,HorowitzEngland17} uses a lot of energy 
efficiently, while our chain $\bar{P}^*$ minimizes energy use by the 
controller, but not the efficiency.  A consequence of this difference is that 
our optimal controlled chain $\bar{P}^*$ 
(\ref{eq:optimization_problem_in_MC_continuous_time}) depends on the 
uncontrolled chain $\bar{Q}$ (see section \ref{sec:minimum_power}), while the 
very fast close-to-optimal chain of 
\cite{HorowitzZhouEngland17,HorowitzEngland17} does not, except in the 
requirement that it be much faster than $\bar{Q}$.

\cite{Todorov06,Todorov09} introduced the KL divergence control setting.  
\cite{Todorov06,Todorov09} uses ``microscopic'' control cost, assigning a cost 
to a trajectory rather than a distribution over trajectories.  Similarly, the 
control goal in \cite{Todorov06,Todorov09} is microscopic (to reach a certain 
subset of the state space $\mc{X}$), rather than macroscopic (to maintain a 
target distribution $\pi^*$ over $\mc{X}$).

\cite{Gopalkrishnan16} considers the problem of erasing a bit of information 
encoded in the stationary distribution of a two-state continuous time Markov 
chain and uses the KL divergence to measure the cost of control, as does our 
work.  Whereas our control goal is to hold a target distribution $\pi^*$ and 
minimize the cost per unit time, \cite{Gopalkrishnan16}'s control goal is to 
have (in the notation of section \ref{sec:MC_continuous_time}) $\mu_T = \pi^*$ 
by a fixed time $T$ and to minimize the total cost used to achieve this.  
Consequently, \cite{Gopalkrishnan16} uses a time-varying controlled chain, 
while ours is constant in time.

\section{Minimum-power controlled chains} \label{sec:minimum_power}

In this section we minimize the power used to hold a desired nonequilibrium 
stationary distribution $\pi^*$. 

\begin{theorem} (Minimum-power chain) \label{thm:I_proj_stat_distrib} Let 
    $\pi^*$ be a distribution on finite set $\mc{X}$ with $\pi^*_i > 0$ for all 
    $i \in \mc{X}$.  Let $\pi(P)$ and $\pi(\bar{P})$ denote the stationary 
    distributions of discrete- and continuous-time chains $P$ and $\bar{P}$, 
    respectively.
    \begin{enumerate}
        \item (Discrete time)
    %(Discrete time minimum power used) 
            Let $Q$ be an irreducible, aperiodic transition probability matrix 
            (the uncontrolled discrete-time chain), let $\pi = \pi(Q)$, and let 
            $P^*$ denote the minimum-power controlled chain with the desired 
            stationary distribution $\pi^*$:
    \begin{equation}
        P^* = \argmin_{P: \pi(P) = \pi^*} D(P||Q) 
        \label{eq:def_I_proj_stat_distrib_disc_time}
    \end{equation}
    where the minimum is over all transition probability matrices with the 
            desired stationary distribution and where $D(P||Q)$ is as defined 
            in (\ref{eq:def_KL_div_MC_discrete_time_1}).  Then $P^*$ exists, is 
            unique, and satisfies for all $i,j\in\mc{X}$:
    \begin{equation}
        P^*_{ij} = Q_{ij} e^{\lambda_i + \eta_j} 
        \label{eq:I_proj_stat_distrib_dic_time_element}
    \end{equation}
    where $(\lambda_i)_{i \in \mc{X}}, (\eta_i)_{i \in \mc{X}}$ are real-valued 
            constants satisfying the recursive relations:
    \begin{align}
        \lambda_i &= - \log\bigg(\sum_{j \in \mc{X}} Q_{ij} e^{\eta_j}\bigg) 
        \label{eq:eta_disc_time} \\
        \eta_i &= \log(\pi^*_i) - \log\bigg(\sum_{j \in \mc{X}} \pi_j Q_{ji} 
        e^{\lambda_j}\bigg) \label{eq:lambda_disc_time}
    \end{align}
\item (Continuous time)
    Let $\bar{Q}$ be a transition rate matrix (the uncontrolled continuous-time 
            chain) with $e^{\bar{Q}}$ irreducible and aperiodic, and let 
            $\bar{P}^*$ denote the minimum-power controlled chain with the 
            desired stationary distribution $\pi^*$:
    \begin{equation}\label{eq:minpowerchain}
        \bar{P}^* = \argmin_{\bar{P}: \pi(\bar{P}) = \pi^*} D(\bar{P}||\bar{Q})
    \end{equation}
    where the minimum is over all transition rate matrices with the desired 
            stationary distribution and where $D(\bar{P}||\bar{Q})$ is as 
            defined in (\ref{eq:KL_div_MC_continuous_time_expression}).  Then 
            $\bar{P}^*$ exists, is unique, and
    satisfies:
    \begin{equation}
        \bar{P}^*_{ij} = \begin{cases}
            \bar{Q}_{ij} e^{\lambda_i - \lambda_j} &: i \neq j \\
            -\sum_{j \in \mc{X}: j \neq i} \bar{P}^*_{ij} &: i = j 
            \label{eq:def_I_proj_stat_distrib_cont_time}
        \end{cases}
    \end{equation}
    where $(\lambda_i)_{i \in \mc{X}}$ are real-valued constants satisfying the 
            recursive relations:
    \begin{align}
        \lambda_i = \frac{1}{2} \log\left(\frac{\sum_{j \in \mc{X}: j \neq i} 
        \pi^*_j \bar{Q}_{ji} e^{\lambda_j}}{\sum_{j \in \mc{X}: j \neq i} 
        \pi^*_i \bar{Q}_{ij} e^{-\lambda_j}}\right) \label{eq:lambda_cont_time}
    \end{align}
    \end{enumerate}
\end{theorem}

Existence and uniqueness of $P^*$ (\ref{eq:def_I_proj_stat_distrib_disc_time}) 
follow as a special case of Lemma 1 of \cite{CsiszarCoverChoi87}.  Existence 
and uniqueness of $\bar{P}^*$ (\ref{eq:def_I_proj_stat_distrib_cont_time}) were 
shown in \cite{BaldiPiccioni98}.  We prove expressions 
(\ref{eq:I_proj_stat_distrib_dic_time_element}) and 
(\ref{eq:def_I_proj_stat_distrib_cont_time}) for $P^*_{ij}$ and 
$\bar{P}^*_{ij}$
%Theorem \ref{thm:I_proj_stat_distrib} 
by setting up a Lagrange multiplier optimization problem, where $\alpha, 
(\lambda_i)_{i \in \mc{X}}, (\eta_i)_{i \in \mc{X}}$ are Lagrange multipliers.  
See Appendix \ref{app:proof} for proof details.

% \begin{proof}
%     For the discrete-time case, we set up an optimization problem with 
%     Lagrangian $\Lambda = D(P||Q) - \sum_{i\in\mc{X}}\lambda_i 
%     \left(\sum_{j\in\mc{X}} \pi^*_i P_{ij} - \pi^*_i\right)$
% \end{proof}

%\textbf{Remarks}:  
The recursive relations (\ref{eq:eta_disc_time}), (\ref{eq:lambda_disc_time}), 
and (\ref{eq:lambda_cont_time}) enable an iterative computation of the chains 
$P^*$ (\ref{eq:def_I_proj_stat_distrib_disc_time})
and $\bar{P}^*$ (\ref{eq:def_I_proj_stat_distrib_cont_time}).  In the 
continuous-time case, for example, we initialize $(\lambda_i^{(0)})_{i \in 
\mc{X}}$ to some value and then use relation (\ref{eq:lambda_cont_time}) to 
compute $\lambda^{(t)}_i$ as a function of $(\lambda^{(t-1)}_i)_{i \in \mc{X}}$ 
at the $t$-th iteration until numerical convergence.  The example in section 
\ref{sec:examples} is computed in this way.

The chains $P^*$ and $\bar{P}^*$ are the I-projections of the chains $P$ and 
$\bar{P}$ on the set of discrete- and continuous-time Markov chains, 
respectively, with a fixed stationary distribution.  The discrete-time case is 
presented in \cite{DemboZeitouniBook,CsiszarCoverChoi87} and the 
continuous-time case in \cite{BaldiPiccioni98}, where $D(P^*||Q)$ and 
$D(\bar{P}^*||\bar{Q})$ arise as large deviations rate functions for the 
empirical marginal distribution.
%the event that the empirical stationary distribution is close to $\pi^*$.  
To our knowledge Theorem \ref{thm:I_proj_stat_distrib} presents the most 
explicit characterization of the I-projection in terms of the Lagrange 
multipliers.

We next specialize our results to the case of reversible uncontrolled Markov 
chains, a case important in equilibrium thermodynamics.  
%This case is important in thermodynamics, corresponding to the time evolution 
%of a closed, or ``undriven'', system.
Let $Q\rev = (Q\rev_{ij})_{i,j\in\mc{X}}$ denote the \textit{time-reverse} of a 
transition probability matrix $Q$.  That is:
\begin{equation}
    Q\rev_{ij} = \frac{\pi_j}{\pi_i} Q_{ji} \label{eq:time_rev_chain_disc_time}
\end{equation}
where $\pi = \pi(Q) = \pi(Q\rev)$ is the stationary distribution.  Then 
$Q\rev_{ij} = \lim_{t \ra \infty} \mb{P}_Q(X_t=j|X_{t+1}=i)$; if $X_0 \sim 
\pi$, then $Q\rev_{ij} = \mb{P}_Q(X_t=j|X_{t+1}=i)$ for all $t$.   Analogously, 
in continuous time, the time-reverse $\bar{Q}\rev$ of a transition rate matrix 
$\bar{Q}$ satisfies $\bar{Q}_{ij} = \frac{\pi_j}{\pi_i} \bar{Q}_{ji}$ for all 
$i,j\in\mc{X}$.  A chain $Q$ is \textit{reversible} if $Q = Q\rev$ 
(analogously, $\bar{Q} = \bar{Q}\rev$ in continuous time).  

\begin{theorem} \label{thm:I_proj_stat_distrib_reversible_base_chain} 
    (Reversible uncontrolled chain) Let notation be as in the statement of 
    Theorem \ref{thm:I_proj_stat_distrib}.
    % Let $Q$ be the uncontrolled transition probability matrix with stationary 
    % distribution $\pi$.  Analogously, in continuous time, let $\bar{Q}$ be 
    % the uncontrolled transition rate matrix with stationary distribution 
    % $\pi$.  Let $\pi^*$ be the desired nonequilibrium stationary 
    % distribution, and let the minimum-power probability matrix $P^*$
    \begin{enumerate}
        \item If the uncontrolled transition probability matrix $Q$ is 
            reversible, then so is the minimum-power chain $P^*$ 
            (\ref{eq:def_I_proj_stat_distrib_disc_time}).  Analogously, in 
            continuous time, if the uncontrolled transition rate matrix 
            $\bar{Q}$ is reversible, then so is $\bar{P}^*$ 
            (\ref{eq:def_I_proj_stat_distrib_cont_time}).
        \item In discrete time, the Lagrange multipliers 
            (\ref{eq:eta_disc_time}), (\ref{eq:lambda_disc_time}) satisfy 
            $\eta_i = \lambda_i + \log(\pi^*_i/\pi_i)$ for all $i \in \mc{X}$.
        \item In continuous time, $\bar{P}^*$ satisfies for $i \neq j$:
            \begin{equation}
                \bar{P}^*_{ij} = \bar{Q}_{ij} \sqrt{\frac{\pi_i \pi^*_j}{\pi_j 
                \pi^*_i}} \label{eq:reversible_cont_time}
            \end{equation}
            and
            \begin{equation}
                D(\bar{P}^*||\bar{Q}) = \frac{1}{2}\sum_{i,j\in\mc{X}:i\neq j} 
                \left(\sqrt{\pi^*_i \bar{Q}_{ij}} - \sqrt{\pi^*_j 
                \bar{Q}_{ji}}\right)^2 \label{eq:KL_div_rev_cont_time}
            \end{equation}
    \end{enumerate}
\end{theorem}
\begin{proof}
    1) We can check that if $Q = Q\rev$, then $D(P||Q) = D(P\rev||Q)$ for all 
    $P$.  Suppose that $P^*$ is not reversible.  Let $P^{\leftrightarrow} 
    \eqdef (P^*+P^{*\text{r}})/2$.  Then $P^{\leftrightarrow}$ is reversible 
    and $\pi(P^{\leftrightarrow}) = \pi(P^*) = \pi^*$.  Since $D(P||Q)$ is 
    strictly convex in $P$, we have $D(P^{\leftrightarrow}||Q) < (D(P^*||Q) + 
    D(P^{*\text{r}}||Q))/2 = D(P^*||Q)$, contradicting the optimality of $P^*$.  
    Another proof: suppose that $P^*$ is not reversible, then $D(P^*||Q) = 
    D(P^{*\text{r}}||Q)$, contradicting the uniqueness of $P^*$ established in 
    Theorem \ref{thm:I_proj_stat_distrib}.  Therefore $P^*$ is reversible.  An 
    analogous argument proves $\bar{P}^*$ is reversible.  2) follows by using 
    time reversal (\ref{eq:time_rev_chain_disc_time}) twice along with the 
    reversibility of $Q$ and $P^*$, established in part 1):
    \begin{align}
        \pi^*_i Q_{ij}e^{\lambda_i + \eta_j} &= \pi^*_i P^*_{ij} = \pi^*_j 
        P^*_{ji} = \pi^*_j Q_{ji} e^{\lambda_j + \eta_i} \\
        &= \pi^*_j \frac{\pi_i}{\pi_j} Q_{ij} e^{\lambda_j + \eta_i}
    \end{align}
    and collecting $i$- and $j$-dependent terms to separate sides of the 
    equality to conclude that $\frac{\pi^*_i}{\pi_i} e^{\lambda_i - \eta_i} = 
    a$ for some constant $a$ for all $i \in \mc{X}$.  Choosing $a=1$ yields the 
    result.

    3) Expression (\ref{eq:reversible_cont_time}) for $\bar{P}^*$ is derived in 
    \cite{BaldiPiccioni98} and $D(\bar{P}^*||Q)$ 
    (\ref{eq:KL_div_rev_cont_time}) is derived as a large deviations rate 
    function in \cite{BertiniFaggionatoGabrielli14}.
\end{proof}

Part 2) of Theorem \ref{thm:I_proj_stat_distrib_reversible_base_chain} lets us 
simplify computation of $P^*$ (\ref{eq:def_I_proj_stat_distrib_disc_time}) 
somewhat when the uncontrolled chain $Q$ is reversible.
%by using $|\mc{X}|$ fewer Lagrange multipliers when the uncontrolled chain $Q$ 
%is reversible than when it is not.

\section{Examples} \label{sec:examples}

% We present a numerical example inspired by the bag-holding thought experiment 
% of Section \ref{sec:introduction}.
% 
% of the minimum-power controlled Markov chain in 
% 

We conclude with numerical examples of minimum-power controlled Markov chains 
with a target stationary distribution.  The first example is a two-state chain 
in discrete time and the second example is a birth-and-death chain in 
continuous and discrete time -- a toy model of the muscle fiber thought 
experiment in the introduction (section \ref{sec:introduction}).

\subsection{Two-state chain in discrete time}

Let $Q$ be a two-state discrete-time Markov chain on set $\mc{X} = \{1,2\}$ and 
let $\pi^* = (\pi^*_1, 1-\pi^*_1)$ be our desired nonequilibrium distribution 
with $\pi^*_1 \in (0,1)$.  All two-state chains are reversible, so we apply 
Theorem \ref{thm:I_proj_stat_distrib_reversible_base_chain} part 2) to compute 
the minimum-power controlled chain with stationary distribution $\pi^*$.  A 
computation shows the $2 \times 2$ minimum-power transition matrix $P^*$ 
(\ref{eq:def_I_proj_stat_distrib_disc_time}) has off-diagonal entries:
    \begin{equation}
        P^*_{ij} = \frac{1}{\pi^*_i} \cdot \frac{1-\sqrt{1-4 \pi^*_1 
        (1-\pi^*_1) s}}{2 s} \ : i \neq j
    \end{equation}
    where the second factor is independent of $i,j$ and
    \begin{equation}
        s \eqdef \frac{1}{Q_{1,2}} + \frac{1}{Q_{2,1}} - \frac{1}{Q_{1,2} 
        Q_{2,1}}
    \end{equation}
        %and $P^*_{2,1} = \frac{\pi^*_1}{1-\pi^*_1} P^*_{1,2}$.  
    The diagonal terms of $P^*$ are such that the row sums are $1$. 
    %for each row.
    %Details of the derivation are given in the Appendix\ref{sec:Appendix}.

\subsection{Birth-and-death chain}

%We conclude with a numerical example of the minimum-power controlled Markov 
%chain with a target stationary distribution.  The example is a toy model of 
%Feynman's muscle fiber thought experiment (see Section 
%\ref{sec:introduction}).  For much more detailed models of molecular motors 
%see \cite{Julicher} and \cite{Qian00}.

We next present the example of the birth-and-death chain as a toy model of 
Feynman's muscle fiber thought experiment (see section \ref{sec:introduction}).  
For detailed models of molecular motors see \cite{Julicher} and \cite{Qian00}.

\subsubsection{Continuous time} \label{sec:cont_time_birth_and_death_chain}

\begin{figure}[!t]
%\capstart
    \centering
    \includegraphics[width=3in]{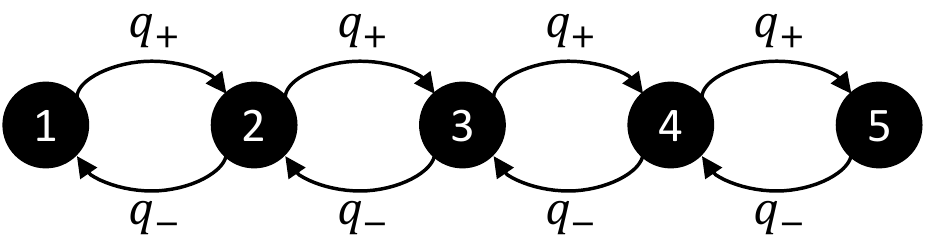}
    \caption{Uncontrolled chain $Q$ (a toy model for the slack arm pulled down 
    by gravity, with smaller-index states closer to the ground and rates $q_- > 
    q_+$): the birth-and-death chain with 5 states in continuous time.  The 
    states correspond to minima in the potential landscape experienced by a 
    single protein at different positions along the muscle fiber.}
    \label{fig:birth_and_death_chain}
\end{figure}

Let $\bar{Q}$ be the continuous-time birth-and-death chain on set 
$\{1,2,\ldots,K\}$ with parameters $q_+, q_- > 0$ depicted in Figure 
\ref{fig:birth_and_death_chain} for $K = 5$.  The chain transitions increments 
from state $i$ to $i + 1$ (resp. decrements to $i-1$) with rate $q_+$ (resp.  
$q_-$).  All other transitions, as well as decrementing from state $1$ and 
incrementing from state $K$, have rate $0$.  $\bar{Q}$ is reversible and its 
stationary distribution is, up to normalization, \cite{LevinPeresWilmerBook}:
\begin{equation}
    \pi_i \sim \left(\frac{q_+}{q_-}\right)^i 
    \label{eq:stat_distrib_birth_and_death_chain}
\end{equation}

Let our control objective be to maintain the target distribution $\pi^*$, a 
geometric distribution on $\{1,\ldots,K\}$:
\begin{equation}
    \pi^*_i \sim b^i \label{eq:target_distribution_birth_and_death_chain}
\end{equation}
where $b > 0$.  Then applying Theorems \ref{thm:I_proj_stat_distrib} and 
\ref{thm:I_proj_stat_distrib_reversible_base_chain} we find the minimum-power 
controlled chain $\bar{P}^*$ (\ref{eq:minpowerchain}) with stationary 
distribution $\pi^*$ to be another birth-and-death chain with increment and 
decrement rates $p_+^*, p_-^*$:
\begin{equation}
    p_{\pm}^* = (q_+ q_-)^{\frac{1}{2}} b^{\pm \frac{1}{2}} 
    \label{eq:birth_and_death_chain_cont_time_parameters}
\end{equation}
The cost per unit time of this birth-and-death chain $\bar{P}^*$ is 
(\ref{eq:KL_div_rev_cont_time}):
\begin{equation}
    %D(\bar{P}^*||\bar{Q}^*) = \left(\sqrt{q_+} - \sqrt{q_- b}\right)^2 
    %\frac{\sum_{k=0}^{K-1} b^i}{\sum_{k=0}^{K} b^i}
    D(\bar{P}^*||\bar{Q}^*) = \left(\sqrt{q_+} - \sqrt{q_- b}\right)^2 
    \left(\frac{1-b^{K-1}}{1-b^K}\right) 
    \label{eq:birth_and_death_chain_cont_time_cost}
\end{equation}
If $b < 1$, then as $K \ra \infty$, we have $D(\bar{P}^*||\bar{Q}) \ra 
\left(\sqrt{q_+} - \sqrt{q_- b}\right)^2$.

\begin{figure}[!t]
%\capstart
    \centering
    \includegraphics[width=3in]{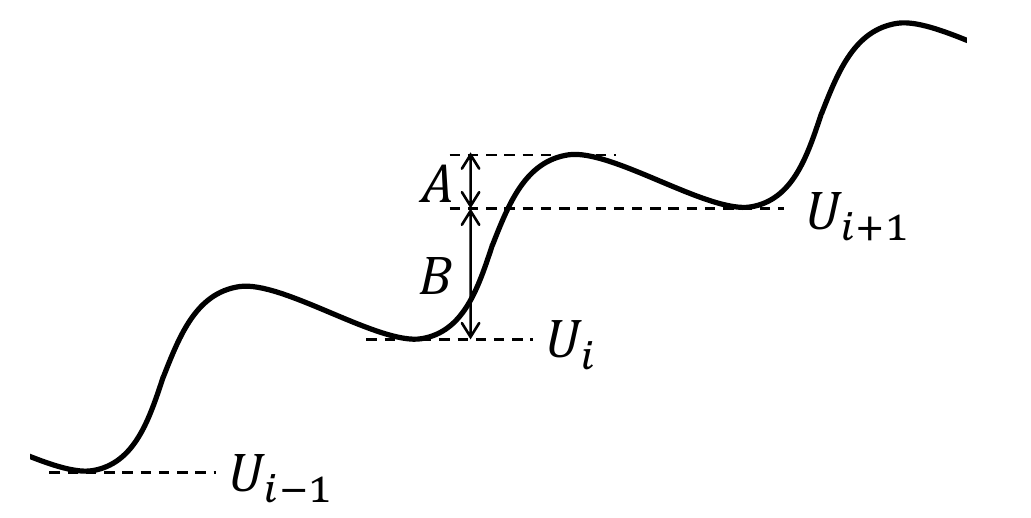}
    \caption{A tilted periodic potential energy landscape.
    %of a protein moving on a muscle fiber in an arm pulled down by gravity.  
    We consider the $i$-th local minimum of the potential as a discrete 
    position state with energy $U_i$ and energy difference $B = U_{i+1} - U_i$ 
    between adjacent states.  The activation energy to go left (resp. right) is 
    $A$ (resp. $A+B$).}
    \label{fig:energy_potential_snapshot}
\end{figure}

Recalling the motivating example of section \ref{sec:introduction}, we can 
think of the birth-and-death chain as a biased random walk, where the random 
walker tends to spend more time at small values of $i$ if $q_- > q_+$.  This is 
a toy model of a myosin protein moving on an actin filament in muscles -- a 
random walk in a tilted energy potential with periodically spaced minima 
corresponding to discrete steps along the fiber.  The energy potential is 
depicted in Figure \ref{fig:energy_potential_snapshot}: a transition from the 
$i+1$-th energy minimum to the $i$-th energy minimum must overcome activation 
energy $A$, and the reverse transition must overcome activation energy $B+A$, 
where $B$ is the energy difference between adjacent states.  In the bag-holding 
thought experiment, $B$ is the gravitational potential energy difference 
between adjacent states, with state $1$ being closest to the ground.

In terms of these energies, the increment and decrement rates $q_+, q_-$ are:
\begin{align}
    q_- &= c e^{-\beta A} \\
    q_+ &= c e^{-\beta (A + B)}
\end{align}
for some constant $c$, where $\beta$ is the inverse temperature.  The 
stationary distribution of the uncontrolled chain $\bar{Q}$ is, up to 
normalization:
\begin{equation}
    \pi_i \sim \left(\frac{q_+}{q_-}\right)^i = e^{-\beta B i}
\end{equation}

Let's write the target nonequilibrium distribution as:
\begin{equation}
    \pi^*_i \sim b^i = e^{-\beta B^* i} 
    \label{eq:target_distribution_cont_time_birth_and_death_chain}
\end{equation}
with energy difference $B^* = (-1/\beta) \log(b)$ between adjacent states.

%If $B > 0$
%This target distribution corresponds to a 

Since the optimal controlled chain $\bar{P}^*$ is another birth-and-death 
chain, we can write its parameters $p^*_+, p^*_-$ 
(\ref{eq:birth_and_death_chain_cont_time_parameters}) in terms of energies 
$A^*$ and $B^*$:
\begin{align}
    p^*_- &= c e^{-\beta A^*} \\
    p^*_+ &= c e^{-\beta (A^* + B^*)}
\end{align}
where, using (\ref{eq:birth_and_death_chain_cont_time_parameters}), we find the 
controlled activation energy $A^*$:
\begin{equation}
    A^* = A + \frac{1}{2}\left(B - B^*\right)
\end{equation}
The cost per unit time of this birth-and-death chain $\bar{P}^*$ is 
(\ref{eq:birth_and_death_chain_cont_time_cost}):
\begin{equation}
    %D(\bar{P}^*||\bar{Q}^*) = c e^{-\beta (A+B)} \left(1 - e^{-\frac{\beta}{2} 
    %(B^* - B)}\right)^2
    D(\bar{P}^*||\bar{Q}^*) = c e^{-\beta A} \left(e^{-\frac{\beta}{2} B} - 
    e^{-\frac{\beta}{2} B^*}\right)^2 \left(\frac{1-e^{-\beta B^* 
    (K-1)}}{1-e^{-\beta B^* K}}\right)
\end{equation}
where if $B^* > 0$, then the last factor tends to $1$ as $K \ra \infty$.
%if $B^* > 0$, then as $K \ra \infty$, we have $D(\bar{P}^*||\bar{Q}) \ra c 
%e^{-\beta A} \left(e^{-\frac{\beta}{2} B} - e^{-\frac{\beta}{2} 
%B^*}\right)^2$.

In the muscle fiber thought experiment (where $B$ is the gravitational 
potential energy difference between adjacent states and state $0$ is closest to 
the ground) if $B^* < B$, then target distribution $\pi^*$ 
(\ref{eq:target_distribution_cont_time_birth_and_death_chain}) corresponds to 
imposing a constant force upwards (away from state $1$) on the random-walking 
myosin protein.  The control objective is macroscopic: rather than control the 
microscopic trajectory of a single myosin protein, we imagine controlling a 
large collection of identical, independent myosin proteins in the same way by 
imposing the same controlled chain $\bar{P}^*$ on all myosins; the bag-holder's 
arm position is determined by an average over the positions of this collection 
of myosins.

\subsubsection{Discrete time} \label{sec:discrete_time_birth_and_death_chain}

Let $Q$ be the discrete-time birth-and-death chain on set $\mc{X} = 
\{1,\ldots,K\}$ with transition probability $q_+$ (resp. $q_-$) to increment 
(resp. decrement) the state from $i$ to $i+1$ (resp. $i-1$).  The stationary 
distribution $\pi$ of $Q$ is as in 
(\ref{eq:stat_distrib_birth_and_death_chain}),
%The stationary distribution $\pi$ of $Q$ is as in 
%(\ref{eq:stat_distrib_birth_and_death_chain}):
%\begin{equation}
%    \pi_i \sim \left(\frac{q_+}{q_-}\right)^i
%\end{equation}
the same as in the continuous time case with transition rates $q_+, q_-$.  Let 
our control objective be to maintain the target distribution $\pi^*_i \sim 
b^i$, a geometric distribution on $\{1,\ldots,K\}$ 
(\ref{eq:target_distribution_birth_and_death_chain}).
%\begin{equation}
%    \pi^*_i \sim b^i 
%    \label{eq:target_distribution_birth_and_death_chain_discrete_time}
%\end{equation}

Then in contrast to the continuous time case of section 
\ref{sec:cont_time_birth_and_death_chain}, the minimum-power controlled chain 
$P^*$ (\ref{eq:def_I_proj_stat_distrib_disc_time}) is not in general a 
birth-and-death chain.  Consider this numerical example: let the target 
nonequilibrium stationary distribution be as in 
(\ref{eq:target_distribution_birth_and_death_chain}) with $b = \sqrt{q_-/q_+}$:
\begin{equation}
    \pi^*_i \sim \sqrt{\pi_{K-i}} \sim \left(\frac{q_+}{q_-}\right)^{-i/2} 
    \label{eq:target_distribution_birth_and_death_chain_discrete_time}
\end{equation}
$\pi^*$ (\ref{eq:target_distribution_birth_and_death_chain_discrete_time}) is 
biased the other way from $\pi_i$, assigning most of its mass to large values 
of $i$ if $q_+ < q_-$.  The square root in 
(\ref{eq:target_distribution_birth_and_death_chain_discrete_time}) makes 
$\pi^*$ look more uniform than $\pi$.

Let $q_- = 0.2$, $q_+ = 0.1$, so that the stationary distribution of $Q$ is 
$\pi_i \sim 2^{-i}$, and $\pi^*_i \sim 2^{i/2}$.  Figure 
\ref{fig:birth_and_death_chain_I_proj} shows the non-zero off-diagonal elements 
of $P^*_{ij}$ with three-digit precision; the increment and decrement 
probabilities vary with state $i$, so $P^*$ is not a birth-and-death chain.

\begin{figure}[!t]
%\capstart
    \centering
    \includegraphics[width=3in]{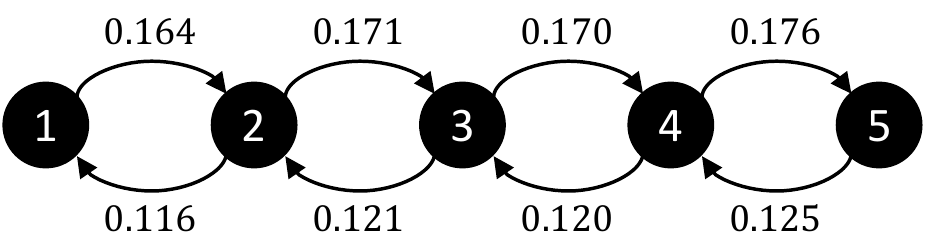}
    \caption{The minimum-power discrete-time controlled chain $P^*$ 
    (\ref{eq:def_I_proj_stat_distrib_disc_time}) for the uncontrolled 
    birth-and-death chain $Q$ with parameters $q_- = 0.2$, $q_+ = 0.1$ and 
    target stationary distribution $\pi^*_i \sim 2^{i/2}$.  Transitions from a 
    state to itself are not shown.}
    \label{fig:birth_and_death_chain_I_proj}
\end{figure}

Finally, Figure \ref{fig:birth_and_death_chain_cost} depicts the time evolution 
of distribution $\mu_t = \mu_0 P^{*t}$ with $\mu_0 = \pi$, and the cost 
$D_{\mu_t}(P^*||Q)$ converging to the minimum power $D(P^*||Q) \approx 0.0315$ 
to maintain the nonequilibrium distribution $\pi^*$.

\begin{figure}[!t]
%\capstart
    \centering
    \includegraphics[width=3.8in]{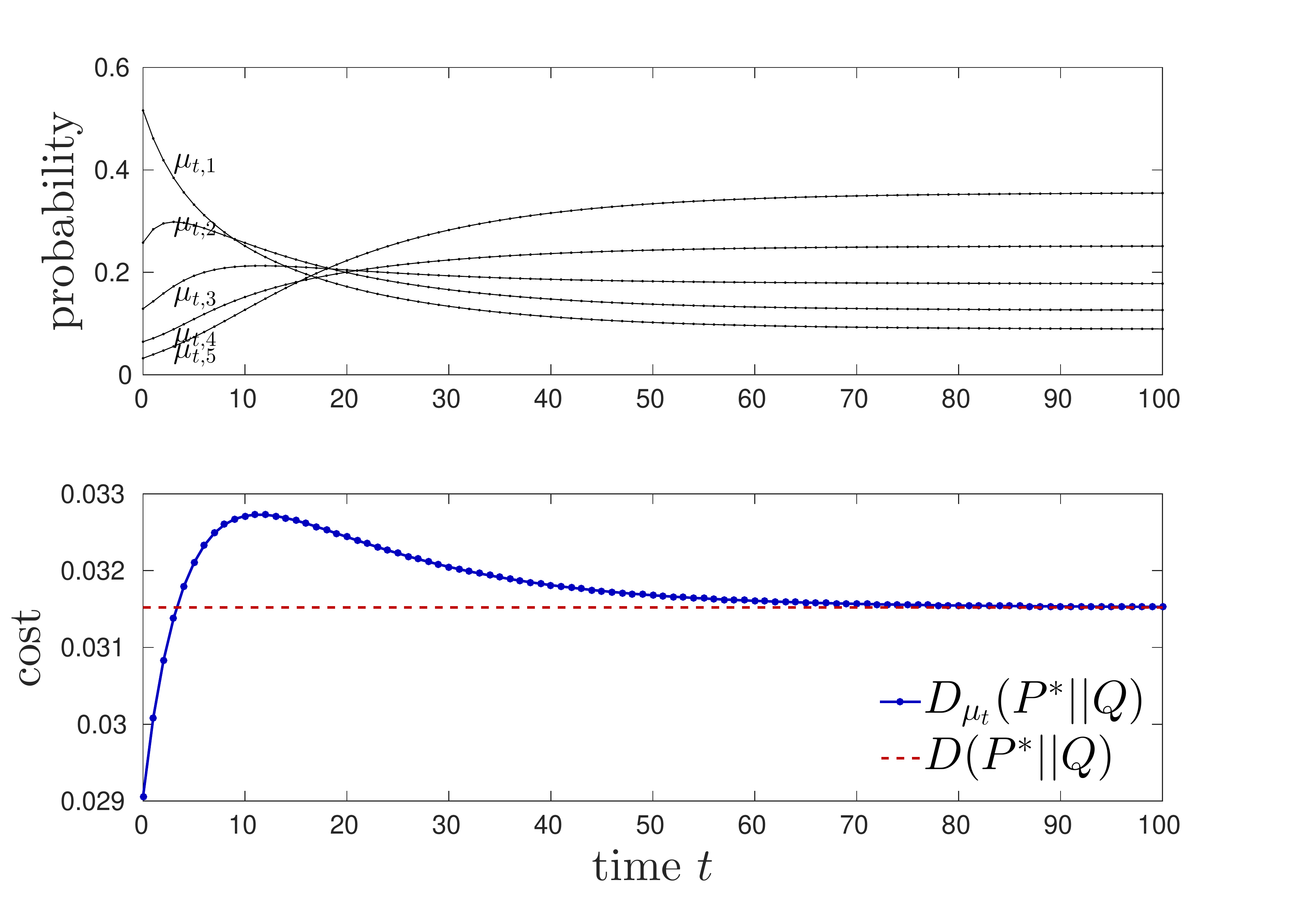}
    \caption{(top) Time evolution of the $|\mc{X}| = 5$ components of $\mu_t = 
    \mu_0 (P^{*})^t$ with $\mu_0 = \pi$, where $P^*$ is the minimum cost chain 
    (\ref{eq:def_I_proj_stat_distrib_disc_time}), showing $\mu_t \ra \pi^*$ as 
    $t \ra \infty$.  (bottom) The cost (blue) $D_{\mu_t}(P^*||Q)$ 
    (\ref{eq:cost_disc_time_KL_div}) and the minimum power (red) $D(P^*||Q)$.}
    \label{fig:birth_and_death_chain_cost}
\end{figure}

Returning to the molecular motor picture, the discrete-time chain $P^*$ differs 
from its continuous-time cousin of section 
\ref{sec:cont_time_birth_and_death_chain} in that $P^*$ is not a 
birth-and-death chain.  Thus an optimal discrete-time control policy modifies 
the base birth-and-death chain $Q$ in a way that depends on state $i$, and so 
can't be thought of as corresponding to a constant, state-independent force 
upwards (away from state $1$) as in the continuous time case.

\section{Discussion}\label{sec:discussion}

This work derives the minimum power required to maintain a target stationary 
distribution given uncontrolled Markov chain dynamics in both discrete and 
continuous time.  We relate KL divergence-like penalties from control theory 
\cite{Todorov06,Todorov09} to the power used to control a Markov process, using 
muscular molecular motors as a guiding example.  The problem of minimizing a KL 
divergence subject to a constrained stationary distribution is familiar from 
large deviations theory 
\cite{DemboZeitouniBook,BaldiPiccioni98,CsiszarCoverChoi87}; the novelty of our 
work is in relating these large deviations results to the thermodynamics of 
``holding'' a distribution, and in computing the minimum-cost chains in some 
important examples: the birth-and-death process in continuous time (a toy model 
for a muscular molecular motor) and two-state chains in discrete time.

%of which muscular molecular motors are a prototype. Although some 
%discretization of the muscular landscape is necessary for application of our 
%theory, our derivation in Section II remains general enough to be applicable 
%to other Markov processes. 

To the best of our knowledge, this is the first time a lower bound on average 
power consumption has been studied in detail for control of the stationary 
distribution.  \cite{HorowitzEngland17} study a related quantity, the minimum 
entropy production rate associated with adding edges (allowing control to 
increase but not decrease transition rates) to a continuous time Markov chain, 
but their notion of cost has the interpretation of energy efficiency, as 
opposed to ours, which is to be interpreted as total energy usage.  
Unsurprisingly, different notions of cost lead to different optimal controlled 
chains: the optimal controlled chain of \cite{HorowitzEngland17} depends on the 
underlying uncontrolled chain only in the requirement that it be much faster, 
while our minimum-cost controlled chain is a function of the uncontrolled 
chain; this function is easy to compute (\ref{eq:reversible_cont_time}) in the 
case of a continuous time, reversible uncontrolled chain, an important case in 
modeling biological processes.

% %While their objective is similarly to push and hold their system in a target 
% %nonequilibrium state, none of their proposed control mechanisms include the 
% %possibility of directly modifying transition rates. We argue that this 
% %scenario should in fact be under consideration, since the rates of many 
% driven %biophysical processes are coupled to those of their accompanying ATP 
% %hydrolysis reactions -- which depend, in part, on the ratio of ATP to ADP 
% %concentrations, which is under some measure of control by the cell. 
% 
% Unsurprisingly, different notions of cost lead to different optimal 
% controllers. By our metric of power-consumption, the entropy-production 
% minimizing controller of \cite{HorowitzEngland17} is not optimal; the rates 
% along the added edges must be infinitely fast compared to the pre-existing 
% ones, and so consume a large amount of power. By \cite{HorowitzEngland17}'s 
% metric of entropy-production, our optimal controller does not saturate the 
% lower bound as it does not go infinitely fast. Our cost, however, has the 
% advantage that the optimal chain is explicit when the uncontrolled chain is 
% reversible, an important case in modeling biological processes.
% 
% %A secondary advantage is its succinct information-theoretic interpretation: 
% %what more natural quantity to use as cost, than the KL divergence between 
% %transition matrices characterizing the relevant Markov processes?

\section*{Acknowledgment}

The authors gratefully acknowledge Hideo Mabuchi for suggestions and insightful discussions.

\bibliographystyle{IEEEtran}
%\bibliography{EEEabrv,refs}
\bibliography{refs}

\appendices

\section{Physical example of KL divergence as energy cost} \label{app:physical_systems}

\begin{figure}[!t]
%\capstart
    \centering
    \includegraphics[width=3in]{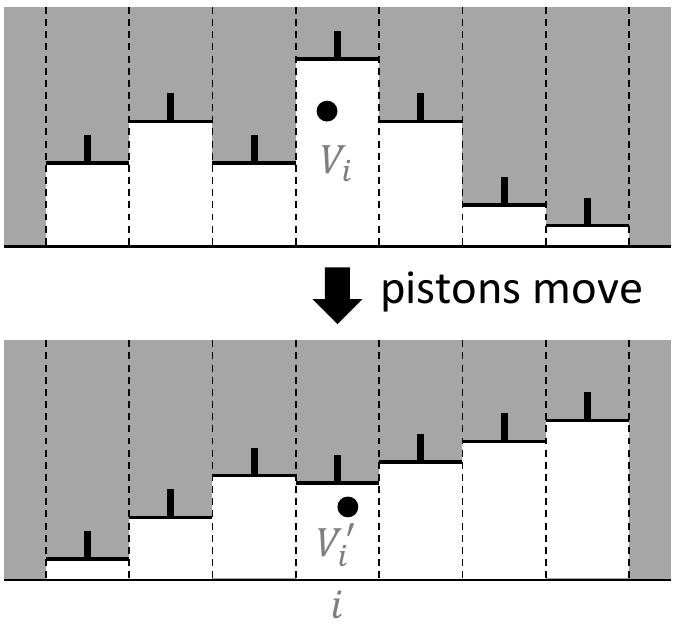}
    %\caption{(top) The positions of the pistons determine the probability 
    %distribution of finding a gas molecule within  Isothermally deforming }
    \caption{(top) A gas molecule is found underneath the $i$-th piston with 
    probability proportional to volume $V_i$. (bottom) Inserting impermeable 
    partitions between the pistons, isothermally compressing to new volumes 
    $V'_i$, and then removing the partitions incurs expected work on the system 
    proportional to the KL divergence between the initial and final probability 
    distributions of the molecule's position.}
    \label{fig:KL_div_work}
\end{figure}

%\begin{figure*}[!t]
%%\capstart
%    \centering
%    %\includegraphics[width=6in]{figures/KL_div_work_v1.pdf}
%    \includegraphics[width=6.5in]{figures/KL_div_work_v3.pdf}
%    %\includegraphics[width=6in]{figures/KL_div_work_v5.pdf}
%    %\includegraphics[width=6in]{figures/KL_div_work_v6.pdf}
%    %\includegraphics[width=6in]{figures/KL_div_work_v7.pdf}
%    \caption{Examples}
%    \label{fig:KL_div_work}
%\end{figure*}

We offer an example of the KL divergence as the cost of sampling from a 
target distribution $p$ given a ``base'' distribution $q$ (see discussion in 
section \ref{sec:KL_div_cost}).

Figure \ref{fig:KL_div_work} presents a slight generalization of the Szilard 
information engine \cite{Szilard29}: molecules of an ideal gas inhabit the 
space formed by movable pistons indexed by $i \in \mc{X}$.  Let $V_i$ denote 
the volume beneath the $i$-th piston and $V \eqdef \sum_{i \in \mc{X}} V_i$ be 
the total volume.  A molecule of gas is equally likely to be found anywhere 
within the space beneath the pistons, corresponding to probability distribution 
$p \eqdef (V_i/V)_{i \in \mc{X}}$ on the pistons.  Now imagine we add 
impermeable partitions between the pistons (vertical dashed lines) and move the 
pistons to new positions (green lines) at constant temperature (perhaps the 
bottom of the box is in thermal contact with a heat reservoir).  The partitions 
prevent mixing between different pistons during compression; we remove them 
afterwards.  Let $V_i'$ be the new volume beneath the $i$-th piston, $V' \eqdef 
\sum_{i \in \mc{X}} V'_i$ be the new total volume, and $q \eqdef (V'_i/V)_{i 
\in \mc{X}}$ be the new piston probability distribution after this deformation.  

What is the work used to perform this deformation?  The work to isothermally 
compress an ideal gas is $k_\text{B} \tau \log(V_i/V'_i)$, where $k_\text{B}$ 
is Boltzmann's constant and $\tau$ is the temperature.  A gas molecule 
occupies the space beneath the $i$-th piston with probability $V_i/V$ before 
compression, so the expected work to move the pistons is
\begin{equation}\label{gaspistons}
    \sum_{i \in \mc{X}} \frac{V_i}{V} k_\text{B} \tau 
    \log\left(\frac{V_i}{V'_i}\right) = k_\text{B} \tau (D(p||q) + \log(V/V'))
\end{equation}
per molecule of gas.  If the new pistons positions are such that the total 
volume is unchanged, then $V' = V$ and the work is proportional to the KL 
divergence $D(p||q)$. We can imagine a sequence of such gas boxes and 
deformations, where the pre-deformation volumes $(V_{t,i})_{i \in \mc{X}}$ at 
time $t$ are determined by drawing a single molecule from volumes 
$(V_{t-1,i})_{i \in \mc{X}}$ at time $t-1$, forming a Markov chain with KL 
divergence control cost.

\section{Proof of Theorem 1} \label{app:proof}

\subsection{Part 1}
We wish to solve the following problem:
\begin{equation}
    P^* \eqdef \argmin_{P: \pi(P)= \pi^*} D(P || Q) \label{eq:def_optimization_problem}
\end{equation}

For ease of manipulation, we work with the empirical joint transition probability distribution $N = (N[i,j])_{i,j\in\mc{X}}$
\begin{equation}
	N[i,j](x_0^n) \eqdef \frac{1}{n}|\{t : (x_{t-1},x_t) = (i,j)\}|.
\end{equation}

We can solve (\ref{eq:def_optimization_problem}) by setting up the Lagrangian: 
\begin{align}
    \Lambda \eqdef & D(P||Q) + \alpha \left(\sum_{i,j \in \mc{X}} N_P[i,j] - 1\right) \nonumber \\
    &+ \sum_{i \in \mc{X}} \lambda_i \left(\sum_{j \in \mc{X}} N_P[i,j] - \pi^*[i]\right) \\
    & + \sum_{j \in \mc{X}} \eta_j \left(\sum_{i \in \mc{X}} N_P[i,j] - \pi^*[j]\right) \label{eq:def_Lagrangian}
\end{align}
where $\alpha$ is a Lagrange multiplier enforcing normalization of the joint transition probability distribution under $P$, $N_P$, and $(\lambda_i)_{i \in \mc{X}}$, $(\eta_i)_{i \in \mc{X}}$ are Lagrange multipliers enforcing the stationary distribution condition $\pi(P) = \pi^*$.  Our solution is a stationary point of the Lagrangian with respect to $N_P$:
\begin{align}
&\left.\frac{\partial \Lambda}{\partial N_P[i,j]} \right|_{N_P = N^*} = \left. \frac{\partial D(P||Q)}{\partial N_P[i,j]} + \alpha + \lambda_i + \eta_j \right|_{N_P = N^*} = 0 \nonumber \\
&\forall i,j \label{eq:Lagrangian_partial_derivative} 
\end{align}
Since $\pi_P = \pi^*$, then $N^*$, $A^*$ satisfy:
\begin{align}
    \left.\frac{\partial D(P||Q)}{\partial N_P[i,j]}\right|_{N_P = N^*} &= \left.\frac{\partial}{\partial N_P[i,j]}\left(D(N_P||N_Q) - D(\pi^*||\pi)\right)\right|_{N_P = N^*} \\
    &= \left.\frac{\partial}{\partial N_P[i,j]}D(N_P||N_Q)\right|_{N_P = N^*} \\
    &= \left.1 + \log\left(\frac{N_P[i,j]}{N_Q[i,j]}\right)\right|_{N_P = N^*} \label{eq:Lagrangian_partial_derivative_expression}
\end{align}
Now using (\ref{eq:def_Lagrangian}), (\ref{eq:Lagrangian_partial_derivative}), and (\ref{eq:Lagrangian_partial_derivative_expression}) we find:
\begin{align}
    N^*[i,j] = N_Q[i,j]e^{-1-\alpha-\lambda_i-\eta_j}
\end{align}
Equivalently
\begin{equation}
    N^* = \text{diag}((e^{-\lambda_i})_{i \in \mc{X}}) \ N_Q \ \text{diag}((e^{-\eta_i})_{i \in \mc{X}}) e^{-1-\alpha}
\end{equation}
Now using the condition $\pi(P) = \pi^*$:
\begin{align}
    \pi^*[i] &= \sum_{j \in \mc{X}} N^*[i,j] \\
    &= e^{-1-\alpha-\lambda_i}\sum_{j \in \mc{X}} N_Q[i,j] e^{-\eta_j}
\end{align}
Therefore
\begin{equation}
    \lambda_i = \log\left(\sum_{j \in \mc{X}} N_Q[i,j] e^{-\eta_j}\right) - \log(\pi^*[i]) - 1 - \alpha \label{eq:recursive_lambda}
\end{equation}
Analogously
\begin{equation}
    \eta_i = \log\left(\sum_{j \in \mc{X}} N_Q[j,i] e^{-\lambda_j}\right) - \log(\pi^*[i]) - 1 - \alpha \label{eq:recursive_eta}
\end{equation}
Now using the normalization condition
\begin{equation}
    1 = \sum_{i,j} N^*[i,j]
\end{equation}
we find
\begin{equation}
    \alpha = \log\left(\sum_{i,j \in \mc{X}} N_Q[i,j]e^{-\lambda_i - \eta_j}\right) - 1 \label{eq:recursive_alpha}
\end{equation}

\subsection{Part 2}
In the continuous case, we work with rate matrices instead of probability transition matrices. We wish to solve the following problem:

\begin{equation}
\begin{aligned}
& \underset{\B{P}: \pi{(\B{P})} = \pi^*}{\text{min}}
& & D(\B{P}||\B{Q}) \\
\end{aligned}
\end{equation}
The Lagrangian is $D(\B{P}||\B{Q})+\sum_i \lambda_i \sum_j \B{P}_{ij} + \sum_j \eta_j \sum_i \pi_i \B{P}_{ij}.$ Differentiating it we get the conditions 
\begin{equation}
      \B{P}_{ij}^* = \begin{cases} \label{eq:rec} 
	\B{Q}_{ij} e^{\eta_i-\eta_j} & i \neq j \\
	-\sum_{j:j\neq i} \B{Q}_{ij} e^{\eta_i-\eta_j} & i = j \\
   \end{cases}
\end{equation}

\end{document}